\documentclass[11pt, english]{article}

\usepackage[margin=1in]{geometry}

\usepackage{varioref}
\usepackage[usenames,svgnames,xcdraw,table]{xcolor}
\definecolor{DarkBlue}{rgb}{0.1,0.1,0.5}
\definecolor{DarkGreen}{rgb}{0.1,0.5,0.1}

\usepackage{nicefrac}

\usepackage{hyperref}
\hypersetup{
	colorlinks   = true,
	linkcolor    = DarkBlue, % color of internal links
	urlcolor     = DarkBlue, % color of external links
	citecolor    = DarkGreen % color of links to bibliography
}

\usepackage{appendix}

\usepackage{mathtools}

\usepackage{algorithm}
\usepackage{array}
\usepackage{algorithmicx}
\usepackage{algpseudocode}

\usepackage[font={small,it}]{caption}

\usepackage{graphicx,epsfig,amsmath,latexsym,amssymb,verbatim}%,revsymb}

\usepackage[square,sort,semicolon,numbers]{natbib}%[round]

\usepackage{comment}
\usepackage{xfrac}
\usepackage{bbm}

\usepackage{verbatim}

\newcommand{\extra}[1]{}

\usepackage{amsmath,amssymb,amsfonts}
\usepackage{amsthm}
\usepackage{thmtools,thm-restate}
\usepackage{enumitem}

\newtheorem{theorem}{Theorem}

\newtheorem{corollary}[theorem]{Corollary}

\newtheorem{lemma}[theorem]{Lemma}

\newtheorem{claim}[theorem]{Claim}

\def\squareforqed{\hbox{\rlap{$\sqcap$}$\sqcup$}}
\def\qed{\ifmmode\squareforqed\else{\unskip\nobreak\hfil
	\penalty50\hskip1em\null\nobreak\hfil\squareforqed
	\parfillskip=0pt\finalhyphendemerits=0\endgraf}\fi}
\def\endenv{\ifmmode\;\else{\unskip\nobreak\hfil
	\penalty50\hskip1em\null\nobreak\hfil\;
	\parfillskip=0pt\finalhyphendemerits=0\endgraf}\fi}

\renewenvironment{proof}{\noindent \textbf{{Proof.~} }}{\qed\medskip}
\newenvironment{proof+}[1]{\noindent \textbf{{Proof #1~} }}{\qed\medskip}

\mathchardef\ordinarycolon\mathcode`\:
\mathcode`\:=\string"8000
\def\vcentcolon{\mathrel{\mathop\ordinarycolon}}
\begingroup \catcode`\:=\active
\lowercase{\endgroup
\let :\vcentcolon
}

\newcommand{\F}{\mathcal{F}}
\newcommand{\I}{\mathcal{I}}
\newcommand{\NSW}{\mathrm{NSW}}

\DeclareMathOperator*{\argmax}{arg\,max}

\newcommand{\BetaVal}{\frac{1}{64n T^2}}
\newcommand{\EpsVal}{\frac{1}{16 n T}}

\allowdisplaybreaks

\title{\bfseries Nash Welfare Guarantees for Fair and Efficient Coverage}

\author{Siddharth Barman\thanks{Indian Institute of Science. {\tt barman@iisc.ac.in}} \quad Anand Krishna\thanks{Indian Institute of Science. {\tt anandkrishna@iisc.ac.in}} \quad Y.~Narahari\thanks{Indian Institute of Science. {\tt narahari@iisc.ac.in}} \quad Soumyarup Sadhukhan\thanks{Indian Institute of Technology Kanpur. {\tt soumyarup.sadhukhan@gmail.com}}}

\date{}

\begin{document}
\maketitle

\begin{abstract}
We study coverage problems in which, for a set of agents and a given threshold $T$, the goal is to select $T$ subsets (of the agents) that, while satisfying combinatorial constraints, achieve fair and efficient coverage among the agents. In this setting, the valuation of each agent is equated to the number of selected subsets that contain it, plus one. The current work utilizes the Nash social welfare function to quantify the extent of fairness and collective efficiency. We develop a polynomial-time $\left(18 + o(1) \right)$-approximation algorithm for maximizing Nash social welfare in coverage instances. Our algorithm applies to all instances wherein, for the underlying combinatorial constraints, there exists an FPTAS for weight maximization. We complement the algorithmic result by proving that Nash social welfare maximization is {\rm APX}-hard in coverage instances. 
\end{abstract}  

 \section{Introduction}
Coverage problems, with a multitude of variants, are fundamental in theoretical computer science, combinatorics, and operations research. These problems capture numerous resource-allocation applications, such as electricity division~\cite{baghel2022fair,oluwasuji2020solving}, sensor allocation~\cite{marden2008distributed}, program testing \cite{kicillof2007achieving}, and plant location~\cite{cornuejols1977exceptional}. 

Coverage problems entail identifying---for a given threshold $T \in \mathbb{Z}_+$ and a set of elements $[n]$---a collection of subsets, $F_1, F_2, \ldots, F_T \subseteq [n]$, that respect particular combinatorial constraints. Here, the problem objective is specified by considering, for each element $i \in [n]$, the number of selected subsets, $F_t$-s, that contain $i$. For instance, in the classic maximum coverage problem \cite{hochbaum1996approximating}, the subsets, $F_1, \ldots, F_T$, are constrained to be from a given set family and the objective is to maximize the number of elements $i \in [n]$ that are contained in at least one of the $F_t$-s, i.e., maximize $| \cup_t F_t|$. 

We study coverage problems where the ground set corresponds to a population of $n$ agents and the cardinal valuation of each agent $i \in [n]$ depends on the number of selected subsets that contain $i$, i.e., the valuation of $i$ depends on the coverage that $i$ receives across the $F_t$-s. Our overarching goal is to select subsets that, while satisfying combinatorial constraints, achieve fair and efficient coverage among the $n$ agents. 

Before detailing the model, we describe a stylized example that illustrates the applicability of the coverage framework. Consider an electricity grid operator tasked with apportioning electricity for $T$ time periods among a set of $n$ agents (consumers with varying electricity requirements). In a time period $t \in [T]$, the total demand of the $n$ agents can exceed the available supply and, hence, the grid operator must select a subset of agents, $F_t \subseteq [n]$, whose electricity consumption can be fulfilled--agents in the subset $F_t$ receive electricity during the $t$th time period and the remaining agents do not. An important desideratum in such load shedding scenarios is to achieve fairness along with economic efficiency; see the motivating work of Baghel et al.~\cite{baghel2022fair} for a thorough treatment of load shedding and its connections with the fair division literature. Indeed, the coverage framework provides an abstraction for this load shedding environment: for each $t \in T$, the selected subset $F_t$ must satisfy a knapsack constraint\footnote{In particular, the total demand of the agents in $F_t$ should be at most the supply at time period $t$.} and the cardinal preference of each agent $i \in [n]$ is captured by the number of subsets that contain $i$, i.e., the number of time periods that $i$ receives electricity.

\paragraph{Combinatorial Constraints.} We study a coverage framework wherein, for each $t \in [T]$, the $t$th selected subset, $F_t \subseteq [n]$, must belong to a set family $\mathcal{I}_t$, i.e., each $\mathcal{I}_t \subseteq 2^{[n]}$ specifies the possible choices for the $t$th selection. Our results do not require the families $\mathcal{I}_t$-s to be given explicitly as input. Our results hold for any $\mathcal{I}_t$-s that admit a fully polynomial-time approximation scheme (FPTAS) for the weight maximization problem: given weights $w_1, \ldots, w_n \in \mathbb{R}_+$, for the $n$ agents, find $\argmax_{X \in \mathcal{I}_t} \ \sum_{i \in X} w_i$.  

For instance, if each $\I_t$ contains the subsets that satisfy a knapsack constraint, then an FPTAS for weight maximization is known to exist \cite{vazirani2001approximation}; in such a case weight maximization corresponds to the standard knapsack problem.\footnote{Recall that in the electricity division example, the subsets $F_t$-s had to satisfy knapsack constraints.} Furthermore, if the families $\I_t$-s are independent sets of matroids, then one can exactly solve the weight maximization problem in polynomial time \cite{schrijver2003combinatorial}. It is relevant to note that matroids provide an expressive construct for numerous combinatorial constraints, e.g., cardinality and partition constraints. Hence, the coverage framework with matroids provides, by itself, an encompassing class of instances. In addition, our result applies to settings in which $\mathcal{I}_t$-s enforce matching constraints: say, for each $t \in [T]$, we have $k_t \leq n$ slots that have to be integrally assigned. Here, each agent prefers a subset of the slots and requires at most one slot for every $t$. Our result holds under such matching constraints, since here weight maximization can be optimally solved in polynomial time via a maximum-weight matching algorithm.\footnote{One can alternatively consider this matching setting as a transversal matroid.} Also, in instances wherein the sizes of the families $\mathcal{I}_t$-s are polynomially large, weight maximization can be efficiently solved by direct enumeration.  

\paragraph{Agents' Valuations.} As mentioned previously, we address settings in which each agent's valuation depends on the number of times it is covered among the selected subsets $F_t$-s. Specifically, for a solution $\mathcal{F} = (F_1,\ldots, F_T) \in \mathcal{I}_1 \times \ldots \times \mathcal{I}_T$, agent $i$'s valuation is defined as $v_i(\mathcal{F}) \coloneqq | \{t\in [T] : i \in F_t \} | + 1$. Note that the valuation of each agent is smoothed by adding $1$. This smoothing enables us to achieve meaningful (multiplicative) approximation guarantees by shifting the valuations and, hence, the collective welfare away from zero. We also note that valuation smoothing has been considered in prior works in fair division; see, e.g., \cite{fain2018fair}, \cite{fluschnik2019fair}, and \cite{kellapproximations21}. 

\paragraph{Nash Social Welfare.} With the overarching aim of achieving fairness along with economic efficiency in coverage instances, we address the problem of maximizing Nash social welfare (NSW). This welfare function is defined as the geometric mean of agents' valuations and it achieves a balance between the extremes of social welfare (a well-studied objective for economic efficiency) and egalitarian welfare (a prominent fairness notion). NSW stands as a fundamental metric for quantifying the extent of fairness in numerous resource-allocation contexts; indeed, in recent years, NSW has been extensively studied in  the fair division literature; see, e.g., \cite{caragiannis2019unreasonable,barman2018finding,halpern2020fair,rosenfeld2021should,li2022constant} and many references therein.     

Nash social welfare satisfies key fairness axioms, including scale freeness, symmetry, and the Pigou-Dalton transfer principle \cite{Moulin03}. The Pigou-Dalton principle requires that the collective welfare should increase under a bounded transfer of value from a well-off agent $i$ to a worse-off agent $j$. NSW satisfies this principle, since the geometric mean of a more balanced valuation profile (of the $n$ agents) is higher than that of a skewed one. At the same time, if the increase in agent $j$'s value is significantly less than the drop experienced by $i$, then NSW does not increase. That is, NSW prefers solutions\footnote{In the current context, a solution is a collection of $T$ subsets $F_1, \ldots, F_T$ that are contained in the underlying set families $\mathcal{I}_1, \ldots, \mathcal{I}_T$, respectively.} that have reduced inequality and, simultaneously, it accommodates for economic efficiency.    

Furthermore, in various fair division contexts, prior works have shown that a solution that maximizes NSW satisfies additional fairness properties, e.g., \cite{caragiannis2019unreasonable}, \cite{fain2018fair}, \cite{halpern2020fair}, \cite{babaioff2021fair}, and \cite{garg2021fair}. Critically, the fact that Nash optimal solutions bear additional guarantees does not undermine the relevance of finding solutions with as high a Nash social welfare as possible. NSW cardinally ranks the solutions and, conforming to a welfarist perspective, one prefers solutions with higher NSW. Therefore, developing approximation guarantees for NSW maximization is a well-justified objective in and of itself. 

\subsection{Our Results and Techniques} 
We develop a constant-factor approximation algorithm for maximizing Nash social welfare in fair coverage instances. Given a set of $n$ agents and threshold $T \in \mathbb{Z}_+$,  our algorithm (Algorithm \ref{Alg}) computes in polynomial time a solution $\mathcal{F} = (F_1, \ldots, F_T) \in \mathcal{I}_1 \times \ldots \times \mathcal{I}_T$ whose Nash social welfare, $\NSW(\F) = \left( \prod_{i=1}^n v_i(\F) \right)^{\frac{1}{n}}$, is at least $\frac{1}{18 + o(1)}$ times the optimal (Theorem \ref{theorem:main}). As mentioned previously, the algorithm only requires blackbox access to an FPTAS for weight maximization over the set families $\I_1, \ldots, \I_T \subseteq 2^{[n]}$.
 
The algorithm starts with an arbitrary solution and iteratively performs updates till it essentially reaches a local maximum of the log social welfare $\varphi(\F) \coloneqq \sum_{i=1}^n \log \left( v_i(\F) \right)$. Here, for any solution $\mathcal{F} = (F_1,\ldots,F_T)$, a local update corresponds to replacing---for some $\tau\in [T]$---the subset $F_\tau$ with some other subset $A_\tau \in \mathcal{I}_\tau$. The algorithm performs the local updates by invoking, as a subroutine, the FPTAS for weight maximization.  

It is relevant to note that while the algorithm is simple in design, its analysis entails novel insights. In particular, the domain of solutions, $\mathcal{I}_1 \times \ldots \times \mathcal{I}_T$, is combinatorial and, hence, it is not obvious if a local maximum solution of $\varphi$ upholds any global approximation guarantees for $\varphi$, let alone for NSW. Furthermore, a multiplicative approximation bound for $\varphi$ does not translate into a multiplicative guarantee for NSW: for any solution $\mathcal{F}$, we have $\frac{1}{n} \varphi \left(\mathcal{F} \right) = \log \left( \NSW (\mathcal{F}) \right)$. Therefore, even though a solution that (globally) maximizes $\varphi$ also maximizes NSW, multiplicative approximation guarantees get exponentially worse when one moves from $\varphi$ to NSW.\footnote{This observation also implies that one cannot directly utilize the approximation guarantee known for the so-called concave coverage problem \cite{barman2021tight} to obtain a commensurate approximation ratio for NSW maximization.}  

Interestingly, in lieu of developing local-to-global approximation guarantees, we rely on counting arguments to establish the approximation ratio. We prove that, at a local maximum solution $\F$ (of the function $\varphi$) and for any integer $\alpha \geq 4$, the number of $\alpha$-suboptimal agents is at most $n/\alpha$; here, an agent $i$ is said to be $\alpha$-suboptimal iff $i$'s current valuation $v_i(\F)$ is (about) ${1}/{\alpha}$ times less than her optimal valuation. We complete the analysis by proving that these Markov-like bounds ensure that the computed solution $\F$ achieves an $\left(18 + o(1)\right)$-approximation guarantee for NSW maximization. 

In addition, we complement the algorithmic result by proving that, in fair coverage instances, NSW maximization is {\rm APX}-hard (Theorem \ref{theorem:apx-hard}). This inapproximability result rules out a polynomial-time approximation scheme (PTAS) for NSW maximization in fair coverage instances. 

\subsection{Additional Related Work and Applications} 
The coverage framework generalizes the well-motivated setup of public decision making \cite{conitzer2017fair}, albeit for agents that have binary additive valuations. The public decision making setup captures settings wherein decisions have to be made on $T$ social issues, that can impact many of the $n$ agents simultaneously. Specifically, each issue $t \in [T]$ is associated with a set of alternatives $A_t = \{a^1_t, a^2_t,\ldots, a_t^{\ell_t} \}$ and every agent $i\in [n]$ has an additive valuation over the issues. That is, for any outcome $\mathcal{A}=(a_1,a_2,\ldots,a_T) \in A_1 \times A_2 \times \ldots \times A_T$, agent $i$'s utility is $u_i(\mathcal{A}) = \sum_{t=1}^T u_i^t(a_t)$; here $u^t_i(a_t) \in \mathbb{R}_+$ is the utility that $i$ gains from the alternative $a_t \in A_t$. 

Indeed, for agents $i \in [n]$ with binary additive valuations (i.e., $u^t_i(a) \in \{0,1\}$ for all $t$ and $a \in A_t$) the coverage framework generalizes public decision making: for every $t \in [T]$, define the set family $\mathcal{I}_t$ by including in it the set $F_a  \coloneqq \{ i \in [n] : u_i^t(a)  = 1 \}$ for each $a \in A_t$. In particular, $\mathcal{I}_t$ contains a set $F_a$, for each alternative $a \in A_t$, where $F_a$ is the set of agents that value alternative $a$. This reduction gives us set families of polynomial size ($|\mathcal{I}_t| = |A_t|$) and, hence, our results specialize  to this case. 

In the public decision making context, Conitzer et al.~\cite{conitzer2017fair} obtain fairness guarantees in terms of relaxations of proportionality. They also show that Nash optimal solutions bear particular fairness properties. Complementing these results and for agents with (smoothed) binary additive valuations, the current work obtains approximation guarantees for NSW in public decision making. %Fain et al.~\cite{fain2018fair} extend public decision making to indivisible public goods and obtain fairness guarantees in terms of the core. Fain et al.~\cite{fain2018fair} also study the fairness properties of solutions that maximize NSW.

The coverage framework also encompasses the standard fair division setting that entails allocation of $m$ indivisible goods among $n$ agents that have binary additive valuations. Multiple prior works have studied NSW in this discrete fair division setting; see, e.g., \cite{barman2018greedy} and \cite{halpern2020fair}. Here, each agent $i \in [n]$ prefers a subset of the goods $V_i \subseteq [m]$ and agent $i$'s valuation $u_i(S) = |S \cap V_i|$, for any $S \subseteq [m]$. One can express this setting as a coverage instance by considering $T = m$ set families each comprised of singleton subsets. Specifically, for each good $g \in [m]$, we have a set family $\mathcal{I}_g$ that includes all singletons $\{i\}$ with the property that $g \in V_i$, i.e., subset $\{i\}$ is included in $\mathcal{I}_g$ iff agent $i$ values good $g$. As in the public decision making setting, here we obtain a coverage instance with polynomially large $\mathcal{I}_t$-s.   

With Nash welfare as a notion of fairness, Fluschnik et al.~\cite{fluschnik2019fair} study fair selection of indivisible goods under a knapsack constraint.\footnote{Fluschnik et al.~\cite{fluschnik2019fair} also highlight connections between NSW and proportional approval voting.} By contrast, the current work addresses combinatorial constraints over subsets of agents.

\section{Notation and Preliminaries}
An instance of a fair coverage problem is specified as a tuple $\langle [n], T, \{\mathcal{I}_t\}_{t=1}^{T}\rangle$, where $[n] = \{1,2,\ldots,n\}$ denotes the set of agents and $T \in \mathbb{Z}_+$ denotes the number of subsets (of the agents) to be selected. Here, for each $t \in [T]$, the $t$th selected subset (say $F_t \subseteq [n]$) is constrained to be from the family $\mathcal{I}_t$, i.e., each $\mathcal{I}_t \subseteq 2^{[n]}$ specifies the possible choices for the $t$th selection. It is not necessary that the set families $\mathcal{I}_t$-s are given explicitly; our algorithmic result only requires a blackbox access to an FPTAS for weight maximization over $\I_t$-s. 

For a fair coverage instance $\langle [n], T, \{\mathcal{I}_t\}_{t=1}^{T}\rangle$, a solution $\mathcal{F} = (F_1,F_2,\ldots, F_T)$ is a tuple with the property that $F_t \in \mathcal{I}_t$ for all $t \in [T]$. We address settings wherein the valuation of each agent depends on the number of times it is covered among the selected subsets. Specifically, for a solution $\mathcal{F} = (F_1,F_2,\ldots, F_T)$, the coverage value $v_i(\mathcal{F})$, of agent $i \in [n]$, is defined as $v_i(\mathcal{F}) \coloneqq | \{t\in [T] : i \in F_t \} | + 1$. Note the coverage value of each agent is smoothed by adding $1$. This smoothing ensures that the Nash social welfare of any solution is nonzero. We, in fact, show that if each agent's value is equated to exactly the number of times it is covered among the subsets, then one cannot achieve \emph{any} multiplicative approximation guarantee for Nash social welfare maximization (Appendix \ref{appendix:non-smooth}).    

The Nash social welfare (NSW) of a solution $\mathcal{F}$ is defined as the geometric mean of the agents' coverage values, $\textsc{NSW}\left(\mathcal{F} \right) \coloneqq  \left( \prod\limits_{i=1}^n v_i(\mathcal{F}) \right)^\frac{1}{n}$. We will write $\mathcal{F}^* = (F_1^*, F_2^*, \ldots, F_T^*)$ to denote a solution that maximizes the Nash social welfare in a given fair coverage instance. Furthermore, a solution $\widehat{\mathcal{F}}$ is said to achieve a $\gamma$-approximation guarantee for the Nash social welfare maximization problem iff $\textsc{NSW}(\widehat{\mathcal{F}}) \geq \frac{1}{\gamma}\textsc{NSW}\left(\mathcal{F}^*\right)$. The current work develops a constant-factor approximation algorithm for NSW maximization in fair coverage instances. 

%does not require the set families $\mathcal{I}_t$-s to be explicitly given as input. Our algorithm

As mentioned previously, the algorithm works with a blackbox access to an FPTAS for weight maximization over $\I_t$-s. Specifically, with parameter $\beta \coloneqq \BetaVal$, we will write \textsc{ApxMaxWt} to denote a subroutine (blackbox) that takes as input weights $w_1, \ldots,w_n \in \mathbb{R}_+$, along with index an $t \in [T]$, and finds a $(1 - \beta)$-approximation to $\max\limits_{X\in \I_t} \ \sum\limits_{i\in X} w_i$. The assumption that weight maximization over $\I_t$-s admits an FPTAS implies that a $(1 - \beta)$-approximation (with $\beta = \BetaVal$) can be computed in polynomial time.     

%Note that if each $\I_t$ contains the subsets that satisfy a knapsack constraint, then such a subroutine (with the specified approximation ratio) is known to exist; in such a case weight maximization corresponds to the standard knapsack problem. Furthermore, if the families $\I_t$-s are independent sets of matroids, then one can (exactly) solve the weight maximization problem in polynomial time. Also, in instances wherein the sizes of the families $I_t$-s are polynomially large, weight maximization can be efficiently solved by direct enumeration.  

For any solution $\mathcal{F} = (F_1,\ldots,F_T)$, index $t \in [T]$, and subset $X \in \mathcal{I}_t$, write $(X, \mathcal{F}_{-t})$ to denote the solution obtained by replacing $F_t$ with $X$, i.e., $(X, \mathcal{F}_{-t}) \coloneqq (F_1,\allowbreak,\allowbreak \ldots, \allowbreak F_{t-1},\allowbreak X, \allowbreak F_{t+1},\ldots,F_T)$. Finally, we will write $\varphi(\mathcal{F})$  to denote the log social welfare of the agents under solution $\mathcal{F}$, i.e., $\varphi(\mathcal{F})  \coloneqq \sum\limits_{i=1}^n \log \left( v_i(\mathcal{F}) \right)$.

\section{Approximation Algorithm for Nash Social Welfare}
This section develops an $(18 + o(1))$-approximation algorithm for maximizing Nash social welfare in fair coverage instances. Given any instance $\langle [n], T, \{\mathcal{I}_t\}_{t=1}^{T}\rangle$, our algorithm \textsc{Alg} (Algorithm \ref{Alg}) starts with an arbitrary solution $\mathcal{F} = (F_1,\ldots,F_T) \in \mathcal{I}_1 \times \ldots \times \mathcal{I}_T$ and iteratively performs local updates as long as it experiences a sufficient (additive) increase in the log social welfare $\varphi$. Here, for any solution $\mathcal{F} = (F_1,\ldots,F_T)$, a local update corresponds to replacing---for some $\tau\in [T]$---the subset $F_\tau$ with some other subset $A_\tau \in \mathcal{I}_\tau$. For updating a solution $\mathcal{F}$ and with $\varphi$ as a guiding objective, the algorithm addresses the problem of finding, for every $t \in [T]$, a subset $A_t \in \mathcal{I}_t$ that achieves $\max\limits_{X \in \mathcal{I}_t} \ \varphi(X, \mathcal{F}_{-t}) - \varphi(\mathcal{F})$. Notably, we reduce this problem to that of weight maximization over $\mathcal{I}_t$-s, by setting appropriate weights $w^t_i$, for each agent $i \in [n]$ and each index $t \in [T]$.  In particular, for a current solution $\mathcal{F}=(F_1, \ldots, F_T)$, the algorithm sets the weights as follows 
\begin{align*}
w^t_i = \begin{cases}
		\log \left(v_i (\mathcal{F}) \right) \ - \ \log \left(v_i (\mathcal{F}) -1 \right) \ & \text{  if } i \in F_t \\
		 \log\left(v_i (\mathcal{F}) + 1 \right)  \ - \ \log\left( v_i(\mathcal{F}) \right) & \text{  otherwise, if } i \in [n] \setminus F_t.
	 \end{cases}
\end{align*} 

We note that for each agent $i \in F_t$, the coverage value $v_i(\mathcal{F}) \geq 2$; this follows from the inclusion of `$+1$' in the definition of coverage value. Hence, the weights (specifically, the terms $\log \left(v_i (\mathcal{F}) -1 \right)$ for $i \in F_t$) are well defined. This is a relevant implication of smoothing the coverage values. 

Moreover, this weight assignment ensures that, for every subset $X \subseteq [n]$, its weight $\sum_{i \in X} w^t_i = \left( \varphi(X, \mathcal{F}_{-t}) - \varphi(\mathcal{F})  \right) + \sum_{j \in F_t} w^t_j$ (see Claim \ref{claim:WtPhi}). Since the weight of the current subset $F_t$ (i.e., $\sum_{j \in F_t} w^t_j$) is fixed, finding a subset $X \in \mathcal{I}_t$ with maximum possible weight is equivalent to finding a subset that maximizes $\varphi(X, \mathcal{F}_{-t}) - \varphi(\mathcal{F})$. In fact, we show that an FPTAS for this weight maximization suffices. As mentioned previously, we denote by $\textsc{ApxMaxWt}(t, w^t_1, \ldots, w^t_n)$ a subroutine (blackbox) that takes as input weights $w^t_1, \ldots, w^t_n \in \mathbb{R}_+$ and finds a $(1 - \beta)$-approximation to $\max\limits_{X\in \I_t} \ \sum\limits_{i\in X} w^t_i$; where the parameter $\beta = \BetaVal$. 

Hence, for updating solution $\mathcal{F}=(F_1, \ldots, F_T)$, the algorithm invokes $\textsc{ApxMaxWt}$ to obtain candidate subsets $A_1, A_2, \ldots, A_T$. If, for some index $\tau \in [T]$, replacing $F_\tau$ by $A_\tau$ leads to a sufficient additive increase $\varphi$, then \textsc{Alg} updates the solution to $(A_\tau, \mathcal{F}_{-\tau})$. Specifically, the algorithm sets parameter $\varepsilon \coloneqq \EpsVal$ and if $\varphi \left( A_\tau, \mathcal{F}_{- \tau} \right)  - \varphi(\mathcal{F}) \geq \frac{\varepsilon n}{8 T}$, then it updates the solution (see Lines \ref{line:phidiffcheck} and \ref{line:update} in Algorithm \ref{Alg}). Otherwise, if for all the candidate subsets the increase in $\varphi$ is less than $\frac{\varepsilon n}{8 T}$, the algorithm terminates. 

Note that, for any solution $\widehat{\mathcal{F}}$, the log social welfare $\varphi(\widehat{\mathcal{F}})$ is at most $n \log (T+1)$.\footnote{Indeed, for any solution $\widehat{\mathcal{F}}$, we have $v_i(\widehat{\mathcal{F}}) \leq T+1$, for all agents $i \in [n]$.} This observation, and the fact that in every iteration of $\textsc{Alg}$ the log social welfare of the maintained solution increases by at least $\frac{\varepsilon n}{8 T}$, imply that the algorithm terminates in polynomial time (Lemma~\ref{lemma:runtime}). Overall, the algorithm efficiently finds a local maximum of $\varphi$. %It is relevant to note that while the algorithm is simple in design, its analysis entails novel insights. In particular, the domain of solutions, $\mathcal{I}_1 \times \ldots \times \mathcal{I}_T$, is combinatorial and, hence, is it not obvious if a local maximum solution of $\varphi$ upholds any global approximation guarantees for $\varphi$, let alone for NSW. Furthermore, observe that a multiplicative approximation bound for $\varphi$ does not translate into a multiplicative guarantee for NSW.\footnote{For any solution $\mathcal{F}$, we have $\frac{1}{n} \varphi \left(\mathcal{F} \right) = \log \left( \NSW (\mathcal{F}) \right)$. Therefore, even though a solution that (globally) maximizes $\varphi$ also maximizes NSW, multiplicative approximation guarantees get exponentially worse when one moves from $\varphi$ to NSW.}

\floatname{algorithm}{Algorithm}
\begin{algorithm}[ht]
	\caption{\textsc{Alg}} \label{Alg} 
	\textbf{Input:} Instance $\langle [n], T, \{\mathcal{I}_t\}_{t=1}^{T} \rangle$. \\ \textbf{Output:} A solution $\mathcal{F} = (F_1, \ldots, F_T)$.
	\begin{algorithmic}[1]
		\State Initialize $\mathcal{F} = (F_1, F_2, \ldots, F_T) \in \mathcal{I}_1 \times \mathcal{I}_2 \times \ldots \times \mathcal{I}_T$ to be an arbitrary solution and, for all agents $i \in [n]$, set coverage value $v_i = v_i(\mathcal{F})$. Set parameter $\varepsilon \coloneqq \EpsVal$. 
		\State For each $t\in [T]$ and all agents $i \in [n]$, set weight \label{line:weights} %$w^t_i = \log(v_i) - \log(v_i-1)$, for all agents $i\in F_t$, and $w^t_j = \log(v_j+1)-\log{v_j}$, for all agents $j\in [n]\setminus F_t$. 
		\begin{align*}
		w^t_i = \begin{cases}
		 \log v_i - \log(v_i-1) & \text{if } i \in F_t \\
		 \log(v_i+1)-\log{v_i} & \text{if } i \in [n] \setminus F_t.
		\end{cases}
		\end{align*} 
		
		\State For each $t\in [T]$, set $A_t$ = \textsc{ApxMaxWt}$(t, w^t_1, w^t_2,\ldots, w^t_n)$.	
		\While{there exists $ \tau \in [T]$ such that $\varphi \left( A_\tau, \mathcal{F}_{- \tau} \right)-\varphi(\mathcal{F}) \geq \frac{\varepsilon n}{8 T}$} \label{line:phidiffcheck}
		\State Update $\mathcal{F} \leftarrow \left( A_\tau, \mathcal{F}_{- \tau} \right)$, i.e., update $F_\tau \leftarrow A_\tau$. \label{line:update}
		\State For all agents $i \in [n]$, update coverage value $v_i = v_i(\mathcal{F})$. 
		\State For each $t\in [T]$ and all agents $i \in [n]$, set weights $w^t_i$ as in Line \ref{line:weights}.  \label{line:weights1}
		\State Set $A_t$ = \textsc{ApxMaxWt}$(t, w^t_1, w^t_2,\ldots, w^t_n)$ for all $t \in [T]$.  \label{line:Atau}
		\EndWhile \label{SATloop}\\
		\Return solution $\mathcal{F}$
	\end{algorithmic}
\end{algorithm}

We establish the approximation ratio via counting arguments. In the analysis, for each maintained solution $\mathcal{F}$, we consider the agents $i$ whose current coverage value, $v_i(\mathcal{F})$, is sufficiently smaller than their optimal coverage value, $v_i(\mathcal{F}^*)$; recall that $\mathcal{F}^*$ denotes a Nash optimal solution.  In particular, for a solution $\mathcal{F}$ and any integer $\alpha \in \mathbb{Z}_+$, we will write $S_\alpha^\mathcal{F}$ to denote the subset of agents whose coverage value is $\alpha (2.25 + \varepsilon)$ times less than their optimal, where, $\varepsilon = \EpsVal$. Formally, for any $\alpha \in \mathbb{Z}_+$, the set of $\alpha$-suboptimal agents is defined as\footnote{Here, the constant $2.25$ is selected to achieve the desired approximation ratio.} 
\begin{align}
 S^\mathcal{F}_\alpha & \coloneqq \left\{i\in [n] : v_i (\mathcal{F}) < \frac{1}{\alpha (2.25 + \varepsilon)} v_i(\mathcal{F}^*) \right\} \label{eqn:sfalpha}
\end{align} %\footnote{Interestingly, this result utilizes the probabilistic method; see Lemma \ref{lem:expphidiff}.}
First, we prove that, for any solution $\mathcal{F}$ and \emph{any} integer $\alpha \geq 4$, if the number of $\alpha$-suboptimal agents is more than $\frac{n}{\alpha}$, then there necessarily exists a local update that increases $\varphi$ by a sufficient amount (Lemma \ref{lemma:tauexists}). Contrapositively, we obtain that, for the solution finally obtained by \textsc{Alg} and for any $\alpha \geq 4$, the number of $\alpha$-suboptimal agents is at most $n/\alpha$. We complete the analysis by proving that this guarantee ensures that $\textsc{Alg}$ achieves a constant-factor approximation ratio for NSW maximization; more formally, we will establish the following theorem (in Section \ref{section:proof-of-main}).
\begin{theorem}[Main Result]
\label{theorem:main}
Given any fair coverage instance $\langle [n], T, \{\mathcal{I}_t\}_{t=1}^{T} \rangle$, with blackbox access to an FPTAS for weight maximization over $\I_t$-s, \textsc{Alg} (Algorithm~\ref{Alg}) computes---in polynomial time---an $\left(18+ \frac{1}{2nT} \right)$-approximate solution for the Nash social welfare maximization problem.
\end{theorem}

\subsection{Algorithm's Analysis}
\label{section:analysis}

The following claim bounds the change in log social welfare $\varphi$ when a solution is updated. 

\begin{restatable}{claim}{Boundphidiff}
\label{lem:boundphidiff}
For a solution $\mathcal{F} = (F_1,\ldots,F_T)$, let value $v_i \coloneqq v_i(\mathcal{F})$ for all agents $i\in [n]$. Then, for any subset $X \subseteq [n]$ and any index $t \in [T]$, we have 
\begin{align*}
\varphi(X,\mathcal{F}_{-t}) - \varphi(\mathcal{F})\geq \sum_{i\in X} \frac{1}{v_i+1} - \sum_{j\in F_t} \frac{1}{v_j-1}.
\end{align*}
\end{restatable}

The proof of Claim \ref{lem:boundphidiff} is deferred to Appendix \ref{appendix:prop-proofs}. Note that here, for each agent $j \in F_t$, the coverage value $v_j(\mathcal{F}) \geq 2$ and, hence, the subtracted terms, $\frac{1}{v_j-1}$, in the claim are well defined. 

Next, we bound the expected change in $\varphi$ when---for any solution $\mathcal{F}=(F_1, \ldots, F_T)$---we replace $F_t$ by $F^*_t$, for a $t \in [T]$ chosen uniformly at random.    

\begin{restatable}{lemma}{Expphidiff}
\label{lem:expphidiff}
For any solution $\mathcal{F} = (F_1,\ldots,F_T)$ and a Nash optimal solution $\mathcal{F^*} = (F^*_1, \ldots,F^*_T)$,  let values $v_i\coloneqq v_i(\mathcal{F})$ and $ v^*_i\coloneqq v_i(\mathcal{F^*})$, for all agents $i\in [n]$. Then, uniformly sampling index $t$ from the set $[T]$, we obtain 
\begin{align*}
\mathbb{E}_{t \in_R [T]} \ \Big[ \varphi(F^*_t,\mathcal{F}_{-t})-\varphi(\mathcal{F}) \Big] \geq \frac{1}{T}\sum\limits_{i=1}^n \left( \frac{v^*_i-1}{v_i+1} \right) \ - \ \frac{n}{T}.
\end{align*} 
\end{restatable}

\begin{proof}
Invoking Claim~\ref{lem:boundphidiff}, with $X = F^*_t$ for each $t \in [T]$, we obtain 
\begin{align}
\mathbb{E}_{t\in_R [T]} \Big[\varphi(F^*_t,\mathcal{F}_{-t})-\varphi(\mathcal{F}) \Big] 
&\geq \mathbb{E}_{t\in_R [T]}\left[\sum_{i\in F^*_t} \frac{1}{v_i+1} - \sum_{j\in F_t} \frac{1}{v_j-1}\right] \nonumber \\
&=\mathbb{E}_{t\in_R [T]}\left[\sum_{i\in [n]} \mathbbm{1}\{i\in F^*_t\}\frac{1}{v_i+1} - \sum_{j\in [n]:v_j\geq 2} \mathbbm{1}\{j\in F_t\} \frac{1}{v_j-1}\right]\tag{since $v_j\geq 2$, for all $j\in F_t$} \nonumber\\
&=\sum_{i\in [n]} \mathbb{P}\{i\in F^*_t\}\frac{1}{v_i+1} - \sum_{j\in [n]:v_j\geq 2}\mathbb{P}\{j\in F_t\} \frac{1}{v_j-1} \label{ineq:indicator} 
\end{align}
Index $t$ is selected uniformly at random from the set $[T]$. Also, by definition, $v_i^*$ is equal to $1$ plus the number of subsets that contain $i$ in the Nash optimal solution $\mathcal{F}^*=(F^*_1, \ldots, F^*_T)$. Hence, the probability $\mathbb{P}\{i\in F^*_t\} = \frac{v^*_i - 1}{T}$, for all agents $i \in [n]$. Similarly, for the solution $\mathcal{F}=(F_1, \ldots, F_T)$, we have $\mathbb{P}\{j\in F_t\}  = \frac{v_j -1}{T}$, for all $j \in [n]$.  These equations and inequality (\ref{ineq:indicator}) give us 
\begin{align*}
\mathbb{E}_{t\in_R [T]} \Big[\varphi(F^*_t,\mathcal{F}_{-t})-\varphi(\mathcal{F}) \Big] & \geq \sum_{i\in [n]} \frac{v^*_i-1}{T} \cdot \frac{1}{v_i+1} - \sum_{j\in [n]:v_j\geq 2} \frac{v_j-1}{T} \cdot \frac{1}{v_j-1} \\
&\geq \frac{1}{T}\sum_{i\in [n]} \left( \frac{v^*_i-1}{v_i+1}\right) \ - \  \frac{n}{T}.
\end{align*}
The lemma stands proved. 
\end{proof}

Next, we show that if, under a solution $\mathcal{F}$, the number of $\alpha$-suboptimal agents is large, then the log social welfare can be sufficiently increased by replacing $F_\tau$ with $F^*_\tau$, for some $\tau \in [T]$.  Recall that $\mathcal{F}^* = (F^*_1, \ldots, F^*_T)$ denotes a Nash optimal allocation and $S^{\mathcal{F}}_\alpha $ denotes the set of $\alpha$-suboptimal agents under solution $\mathcal{F}$; see equation (\ref{eqn:sfalpha}).  

\begin{restatable}{lemma}{Tauexists}
\label{lemma:tauexists}
For any solution $\mathcal{F}=(F_1,\ldots,F_T)$  and any $\alpha \geq 4$, if the number of $\alpha$-suboptimal agents is at least $\frac{n}{\alpha}$ (i.e., $|S^\mathcal{F}_\alpha|>\frac{n}{\alpha}$), then there exists an index $\tau \in [T]$ such that 
\begin{align*}
\varphi(F^*_\tau, \mathcal{F}_{-\tau})-\varphi(\mathcal{F})\geq \dfrac{\varepsilon n}{2T}. \label{eq:Sexists}
\end{align*}
\end{restatable}
\begin{proof}
Consider any solution $\mathcal{F}$ and integer $\alpha \geq 4$ such that $|S^\mathcal{F}_\alpha|>\frac{n}{\alpha}$. For each agent $i \in [n]$, write $v_i \coloneqq v_i(\mathcal{F})$ and $v^*_i = v_i(\mathcal{F}^*)$. Now, Lemma~\ref{lem:expphidiff} gives us  
\begin{align*}
\mathbb{E}_{t\in_R [T]} \Big[\varphi(F^*_t,\mathcal{F}_{-t})-\varphi(\mathcal{F})\Big] & \geq \frac{1}{T}\sum\limits_{i=1}^n \left( \frac{v^*_i-1}{v_i+1} \right) \ - \ \frac{n}{T}\\
& \geq \frac{1}{T}\sum\limits_{i \in S^\mathcal{F}_\alpha} \left( \frac{v^*_i-1}{v_i+1}\right) \ - \ \frac{n}{T}\\ 
& \geq \frac{1}{T}\sum\limits_{i \in S^\mathcal{F}_\alpha} \left( \frac{\alpha (2.25+\varepsilon) v_i-1}{v_i+1}\right) \ - \ \frac{n}{T} \tag{by definition of $S^\mathcal{F}_\alpha$} 
\end{align*}
Claim~\ref{prop:epsilonineq} (stated and proved in Appendix \ref{appendix:prop-proofs}) shows that $\frac{\alpha (2.25 + \varepsilon) v -1}{v+1} \geq \left(1+\frac{\varepsilon}{2}\right)\alpha$, for all integers $\alpha \geq 4$ and $v \geq 1$. Therefore, the above-mentioned inequality simplifies to 
\begin{align*}
\mathbb{E}_{t\in_R [T]} \Big[\varphi(F^*_t,\mathcal{F}_{-t})-\varphi(\mathcal{F})\Big] & \geq \frac{1}{T}\sum\limits_{i \in S^\mathcal{F}_\alpha} \left(1+\frac{\varepsilon}{2}\right)\alpha \ - \ \frac{n}{T}\\  
& > \frac{1}{T}\frac{n}{\alpha} \left(1+\frac{\varepsilon}{2}\right)\alpha \ - \ \frac{n}{T} \tag{since $|S^\mathcal{F}_\alpha|>\frac{n}{\alpha}$} \\
&  = \frac{\varepsilon n}{2T}.
\end{align*}
Therefore, there exists a $\tau \in [T]$ such that 
\begin{align*}
	\varphi(F^*_\tau,\mathcal{F}_{-\tau})-\varphi(\mathcal{F}) \geq \dfrac{\varepsilon n}{2T}.
\end{align*}
This completes the proof of the lemma. 
\end{proof}

Using Lemma \ref{lemma:tauexists}, we will establish in Corollary \ref{lem:algexecution}, below, that the algorithm continues to iterate as long as the number of $\alpha$-suboptimal agents is more than $n/\alpha$. The proof of the corollary also utilizes the following claim.

\begin{restatable}{claim}{ClaimWtPhi}
\label{claim:WtPhi}
Let $\mathcal{F}=(F_1, \ldots, F_T)$ be any solution considered in \textsc{Alg} (Algorithm~\ref{Alg}) and, for all indices $t \in [T]$ and agents $i \in [n]$, let $w^t_i$-s be the corresponding weights set in Lines  \ref{line:weights} or \ref{line:weights1}. Then, the weight of any subset $X \subseteq [n]$ satisfies 
\begin{align*}
\sum_{i \in X} w^t_i = \left( \varphi(X, \mathcal{F}_{-t}) - \varphi(\mathcal{F})  \right) + \sum_{j \in F_t} w^t_j. 
\end{align*}
\end{restatable}
The proof of this claim appears in Appendix \ref{appendix:prop-proofs}. \\

\begin{corollary}\label{lem:algexecution}
For any solution $\mathcal{F}=(F_1,\ldots,F_T)$ considered in \textsc{Alg} (Algorithm~\ref{Alg}) and any $\alpha \geq 4$, if the number of $\alpha$-suboptimal agents is at least $\frac{n}{\alpha}$ (i.e., $|S^\mathcal{F}_\alpha|>\frac{n}{\alpha}$), then the execution condition in the while-loop (Line \ref{line:phidiffcheck}) of \textsc{Alg} holds.  
\end{corollary}
\begin{proof}
Consider any solution $\mathcal{F}$ in \textsc{Alg} and integer $\alpha \geq 4$ such that $|S^\mathcal{F}_\alpha|>\frac{n}{\alpha}$. In such a case, we will show that there exists an index $\tau \in [T]$ for which the subset $A_\tau$ returned by the subroutine $\textsc{ApxMaxWt}(\tau, w^\tau_1, \ldots, w^\tau_n)$ (in Line~\ref{line:Atau}) satisfies $\varphi(A_\tau,\mathcal{F}_{-\tau})-\varphi(\mathcal{F}) \geq \frac{\varepsilon n}{8T}$. Hence, the while-loop continues to iterate. 

The desired index is in fact the one identified in Lemma \ref{lemma:tauexists}. In particular, Lemma \ref{lemma:tauexists} ensures that for an index $\tau \in [T]$ we have 
\begin{align}
\varphi(F^*_\tau, \mathcal{F}_{-\tau})-\varphi(\mathcal{F}) & \geq \frac{\varepsilon n}{2T} \label{ineq:tauhigh}
\end{align}
Now, Claim \ref{claim:WtPhi} (with $X = F^*_\tau$) gives us 
\begin{align*}
\sum_{i \in F^*_\tau} w^\tau_i & = \left( \varphi(F^*_\tau, \mathcal{F}_{-\tau}) - \varphi(\mathcal{F})  \right) + \sum_{i\in F_\tau} w^\tau_i  \\
& \geq \frac{\varepsilon n}{2T} + \sum_{i\in F_\tau} w^\tau_i \tag{via inequality (\ref{ineq:tauhigh})}
\end{align*}
Therefore, 
\begin{align}
\max_{X \in \I_\tau} \left\{ \sum\limits_{i\in X}  w^\tau_i \right\} & \geq \sum_{i\in F_\tau} w^\tau_i + \frac{\varepsilon n}{2T}\label{eq:lbmaxweight}
\end{align}
Recall that $\textsc{ApxMaxWt}(\tau, w^\tau_1, \ldots, w^\tau_n)$ returns a set $A_\tau \in \mathcal{I}_\tau$ with the property that
\begin{align}
\sum\limits_{i\in A_\tau} w^\tau_i & \geq \left(1- \beta \right)\left(\max_{X\in \I_\tau} \ \sum\limits_{i\in X} w^\tau_i\right)\label{eq:apxmaxweight}
\end{align}
Here, parameter $\beta = \BetaVal$. Since $\varepsilon = \EpsVal$, we have $\beta = \frac{\varepsilon}{4T}$. Inequalities~\eqref{eq:lbmaxweight} and~\eqref{eq:apxmaxweight} give us
		\begin{align*}
			\sum\limits_{i\in A_\tau} w^\tau_i &\geq (1- \beta) \left( \sum_{i\in F_\tau} w^\tau_i + \frac{\varepsilon n}{2T}\right)\\ 
			& = \sum_{i\in F_\tau} w_i^\tau + \frac{\varepsilon n}{2T} - \beta \sum_{i\in F_\tau} w^\tau_i - \frac{\beta \varepsilon n}{2T} \\
			& \geq \sum_{i\in F_\tau} w_i^\tau + \frac{\varepsilon n}{2T}  - \beta n -  \frac{\beta \varepsilon n}{2T} \tag{since $\sum_{i\in F_\tau} w^\tau_i \leq n$} \\
			&  = \sum_{i\in F_\tau} w_i^\tau + \frac{\varepsilon n}{2T}  - \frac{\varepsilon n}{4T} -  \frac{\beta \varepsilon n}{2T} \tag{since $\beta = \frac{\varepsilon}{4T}$} \\
			& \geq \sum_{i\in F_\tau} w_i^\tau + \frac{\varepsilon n}{2T}  - \frac{\varepsilon n}{4T} -  \frac{\varepsilon n}{8T}  \tag{since $\beta \leq \frac{1}{4} $} \\
			& = \sum_{i\in F_\tau} w_i^\tau + \frac{\varepsilon n}{8T}.
\end{align*}
Applying Claim \ref{claim:WtPhi}, with $X = A_\tau$, we get $\varphi(A_\tau,\mathcal{F}_{-\tau})-\varphi(\mathcal{F})\geq\frac{\varepsilon n}{8T}$. Therefore, the execution condition in the while-loop of \textsc{Alg} holds. This establishes the corollary. 
\end{proof}

We conclude the section by showing that the algorithm runs in polynomial time. 

\begin{lemma}[Runtime Analysis]
\label{lemma:runtime}
Given any fair coverage instance $\langle [n], T, \{\mathcal{I}_t\}_{t=1}^{T} \rangle$ with blackbox access to an FPTAS for weight maximization over $\I_t$-s, \textsc{Alg} (Algorithm~\ref{Alg}) terminates in time that is polynomial in $n$ and $T$. 
\end{lemma}
\begin{proof}
For any solution $\mathcal{F}$, the coverage values $v_i(\mathcal{F}) \geq 1$, for agents $i \in [n]$. Hence, for the initial solution (arbitrarily) selected by the algorithm, we have $\varphi(\mathcal{F}) = \sum\limits_{i=1}^n \log(v_i(\F)) \geq 0$. 
In addition, since  the coverage values of the agents under any solution are at most $T+1$, the log social welfare $\varphi$ across all solutions is upper bounded by $n \log (T+1)$. Furthermore, note that in every iteration of $\textsc{Alg}$ the log social welfare of the maintained solution increases additively by at least $\frac{\varepsilon n}{8 T}$. These observations imply that the algorithm terminates after $O\left(  n T^2 \log T\right)$ iterations; recall that $\varepsilon = \EpsVal$. Since each iteration executes in polynomial time, the time complexity of the algorithm is polynomial in $n$ and $T$. The lemma stands proved. 
\end{proof}

\subsection{Proof of Theorem \ref{theorem:main}}
\label{section:proof-of-main}
This section establishes the approximation ratio of \textsc{Alg}. For the given fair coverage instance, let $\F = (F_1,\ldots,F_T)$ be the solution returned by \textsc{Alg} and $\F^* = (F^*_1,\ldots,F^*_T)$ be a Nash optimal allocation. Note that $v_i(\F)\geq 1$ and $v_i(\F^*)\leq T+1$, for all agents $i\in [n]$. Hence, for each agent $i \in [n]$, the following bound holds: $v_i(\F) \geq \frac{1}{T+1}v_i(\F^*)$. 

We partition the set of agents $[n]$ considering the multiplicative gap between the coverage values under $\mathcal{F}$ and $\mathcal{F}^*$. Specifically, for each integer $d\in \left\{2,3,\ldots, \lceil \log (T +1) \rceil \right\}$, define the set 
\begin{align*}
X_{2^d} & \coloneqq \left\{i\in [n] : \frac{1}{2^{d+1}}\frac{v_i(\F^*)}{(2.25+\varepsilon)} \leq v_i(\F) < \frac{1}{2^{d}} \frac{v_i(\F^*)}{(2.25+\varepsilon)}\right\}.
\end{align*} 
Furthermore, write $X' \coloneqq [n] \setminus \left( \bigcup\limits_{d=2}^{\lceil \log{(T+1)}\rceil}X_{2^d} \right)$. Since all agents $i$ satisfy $v_i(\F)\geq \frac{1}{T+1}v_i(\F^*)$, the subset $X'$ only contains agents $j \in [n]$ with the property that $v_j(\F) \geq \frac{1}{4}\frac{v_j(\F^*)}{(2.25+\varepsilon)}$. Also, note that the subsets $X_{2^d}$-s and $X'$ form a partition of the set of agents $[n]$; in particular, $|X'|+\sum_{d\geq 2} |X_{2^d}| = n$. 

Recall that $S^\mathcal{F}_\alpha$ denotes the set of $\alpha$-suboptimal agents (see equation (\ref{eqn:sfalpha})). Also, note that, with $\alpha = 2^d$, we have $X_\alpha \subseteq S^\mathcal{F}_\alpha$. Moreover, by the contrapositive of Lemma \ref{lem:algexecution}, for the solution $\F = (F_1,\ldots,F_T)$, returned by \textsc{Alg}, we have  
\begin{align}
\left|X_{2^d} \right| & \leq \left|S^\mathcal{F}_{2^d} \right| \ \leq \frac{n}{2^d} \qquad \ \text{ for all $2\leq d\leq \lceil \log{(T+1)} \rceil$}\label{eq:boundX}
\end{align}

For any subset of agents $Y \subseteq [n]$, write $\rho(Y) \coloneqq \prod_{i\in Y} \frac{v_i(\F)}{v_i(\F^*)}$, if subset $Y \neq \emptyset$. Otherwise, if $Y = \emptyset$, define $\rho(Y) \coloneqq 1$. 
To bound the approximation ratio of the algorithm, we consider
	\begin{align*}
	\frac{\textsc{NSW}(\F)}{\textsc{NSW}{(\F^*)}} &= \left(\rho(X') \prod\limits_{d=2}^{\lceil \log(T+1) \rceil} \rho(X_{2^d})\right)^\frac{1}{n}\\
	&\geq \left(\left(\frac{1}{9+4\varepsilon}\right)^{|X'|}\prod\limits_{d\geq 2}\rho(X_{2^d})\right)^\frac{1}{n} \tag{$v_j(\F) \geq \frac{1}{4(2.25 + \varepsilon)}v_j(\F^*)$ for all $j\in X'$}\\
	&\geq \left(\left(\frac{1}{9+4\varepsilon}\right)^{|X'|}\prod\limits_{d\geq 2}\left(\frac{1}{2^{d+1}(2.25+\varepsilon)}\right)^{|X_{2^d}|}\right)^\frac{1}{n} \tag{$v_i(\F) \geq \frac{1}{2^{d+1}(2.25+\varepsilon)}$ for all $i\in X_{2^d}$}\\
	& =  \frac{1}{9+4\varepsilon}\left(\prod\limits_{d\geq 2}\left(\frac{1}{2^{d-1}}\right)^\frac{|X_{2^d}|}{n}\right) \tag{since $|X'|+\sum_{d\geq 2} |X_{2^d}| = n$}\\
	&\geq \frac{1}{9+4\varepsilon}\left(\prod\limits_{d\geq 2}\left(\frac{1}{2^{d-1}}\right)^\frac{1}{2^d}\right) \tag{via inequality~\eqref{eq:boundX}}\\
	\end{align*}
Claim \ref{prop:dineq} (proved in Appendix \ref{appendix:prop-proofs}) shows that the product $\prod_{d\geq 2} \left(\frac{1}{2^{d-1}}\right)^\frac{1}{2^d} \geq \frac{1}{2}$. Hence, the stated approximation ratio follows
 \begin{align*}
 \frac{\textsc{NSW}(\F)}{\textsc{NSW}{(\F^*)}} \geq \frac{1}{9+4\varepsilon}\left(\prod\limits_{d = 2}^{\lceil \log (T+1) \rceil}\left(\frac{1}{2^{d-1}} \right)^\frac{1}{2^d}\right) \geq \frac{1}{9+4\varepsilon} \cdot \frac{1}{2} = \frac{1}{18 + 8 \varepsilon}.
 \end{align*}

\section{{\rm APX}-Hardness of Fair Coverage}
\label{section:apx-hard}
This section shows that NSW maximization in fair coverage instances is {\rm APX}-hard.  In particular, we prove that there exists an absolute constant $\gamma >1$ such that it is {\rm NP}-hard to approximate the problem within factor $\gamma$. Hence, a constant-factor approximation is the best one can hope for NSW maximization in fair coverage instances, unless ${\rm P} = {\rm NP}$. The hardness result is obtained via an approximation preserving reduction from the following gap version of the maximum coverage problem. \\

\noindent
{\it Maximum $k$-Coverage} \cite{feige1998threshold}:  Given a universe of elements $U = \{1,2,\ldots,n\}$, a threshold $k \in \mathbb{Z}_+$, and a set family $\mathcal{S} = \big\{S_\ell \subseteq [n]\big\}_{\ell=1}^N$, it is {\rm NP}-hard to distinguish between
\begin{itemize}
	\item {\rm YES} Instances: There exists a collection of $k$ subsets in $\mathcal{S}$ that covers all the elements, i.e., the union of the $k$ subsets is equal to $[n]$.
	\item {\rm NO} Instances: Any collection of $k$ subsets from $\mathcal{S}$ covers at most $\left(1-\frac{1}{e}\right)n$ elements, i.e., the union of any $k$ subsets from $\mathcal{S}$ has cardinality at most $\left(1-\frac{1}{e}\right)n$.
\end{itemize}

This hardness result of Feige \cite{feige1998threshold} holds even for instances that satisfy the following properties: (i) all the subsets in $\mathcal{S}$ have the same size $\tau$, i.e., $|S_\ell| = \tau$ for all subsets $S_\ell \in \mathcal{S}$, and (ii) the threshold $k = n/\tau$. Properties (i) and (ii) will be utilized in our approximation preserving reduction.\footnote{The properties also ensure that in the {\rm YES} case there is a collection of $k = \frac{n}{\tau}$ subsets that are pairwise disjoint and they cover all of $[n]$. That is, in the {\rm YES} case there exists a perfect cover.}

The {\rm APX}-hardness result is established next. Notably, this negative result is applicable even for fair coverage instances in which the set families $\mathcal{I}_t$-s are explicitly given as input. 

\begin{theorem}
\label{theorem:apx-hard}
In fair coverage instances, it is {\rm NP}-hard to approximate the maximum Nash social welfare within a factor of $1.092$. 
\end{theorem}
\begin{proof}
Given an instance of the maximum $k$-coverage problem with universe $U=\{1,2,\ldots,n\}$ and set family $\mathcal{S} =\{S_1,S_2,\ldots,S_N\}$ of $\tau$-sized subsets of $[n]$, we construct a fair coverage instance with $n$ agents and $T = k$. Since threshold $k = \frac{n}{\tau}$, we have $T = \frac{n}{\tau}$. To complete the construction and obtain an instance $\langle [n], T, \{\mathcal{I}_t\}_{t=1}^{T}\rangle$, we set the families $\mathcal{I}_t = \mathcal{S}$, for all $t \in [T]$. 

First, we show that if the underlying maximum coverage instance is a {\rm YES} instance, then the optimal NSW in the constructed fair coverage instance is at least $2$. Note that in the {\rm YES} case there exists a size-$k$ collection  $\mathcal{S}'=\{S'_1, S'_2, \ldots, S'_k \}\subseteq \mathcal{S}$ that covers all of $[n]$. Also, by construction, $T = k$ and $\mathcal{I}_t = \mathcal{S}$ for all $1 \leq t \leq k$. Hence, for each $t \in [T]$, we have $S'_t \in \mathcal{I}_t$. Therefore, the tuple $\mathcal{F}' = (S'_1, S'_2, \ldots, S'_k)$ is a solution under which $v_i(\mathcal{F}') \geq 2$, for all agents $i \in [n]$.\footnote{In fact, for each agent $i$ the coverage value $v_i(\mathcal{F}') =2$, since $i$ is contained in exactly one of the subsets $S'_t$-s. Recall that properties (i) and (ii) ensure that $\mathcal{S}'$ is a perfect cover.}
This bound on the coverage value of the agents implies that in the current case, the optimal Nash social welfare is at least $2$. 

Now, we show that in the {\rm NO} case the optimal NSW is at most $c$, for an absolute constant $c<2$. Here, consider any solution $\F = (F_1,\ldots,F_{T})$ in the constructed fair coverage instance. We have $T = k = \frac{n}{\tau}$ and, by construction, $F_t \in \mathcal{S}$. Furthermore, given that we are in the {\rm NO} case, the collection of subsets $\{F_1, F_2, \ldots, F_T\} \subseteq \mathcal{S}$ covers at most $(1-\frac{1}{e})n$ elements. Let $L$ denote the set of agents not covered by the subsets $F_t$-s and write $\ell \coloneqq |L| \geq \frac{n}{e}$.  Since each agent $i \in L$ is not covered under $\mathcal{F}$, we have $v_i(\mathcal{F})= 1$ for all $i \in L$. Furthermore, note that the agents in the set $L^c \coloneqq [n] \setminus L$ are covered by the $T = k = \frac{n}{\tau}$ subsets $F_1, \ldots, F_T$, and each of these subsets is of size $\tau$. Therefore, 
\begin{align*}
\sum_{j \in L^c}  v_j(\mathcal{F}) & = \sum_{t = 1}^{n/\tau} |F_t| \ + \ |L^c| \\
& = \frac{n}{\tau} \tau \ + \ |L^c| \tag{since $|F_t| = \tau$ for each $t$} \\
& = n + (n-\ell) \tag{$\ell = |L|$}
\end{align*}
Hence, the average social welfare among agents in $L^c$ satisfies $\frac{1}{|L^c|} \sum_{j \in L^c}  v_j(\mathcal{F}) = \frac{2n - \ell}{n - \ell}$. This bound and the AM-GM inequality give us $\prod\limits_{j \in L^c} v_i(\mathcal{F}) \leq \left( \frac{2n - \ell}{n - \ell} \right)^{|L^c|}$.  Therefore, we can bound the Nash social welfare of $\F$ as follows
\begin{align}
\NSW(\F) = \left( \prod\limits_{i \in L} v_i(\F) \ \prod\limits_{j \in L^c} v_j(\F) \right)^{\frac{1}{n}} \leq 1^{\frac{\ell}{n}} \  \left(\frac{2n-\ell}{n-\ell}\right)^{\frac{n-\ell}{n}} = \left(\frac{2 -\ell/n}{1-\ell/n}\right)^{\left( 1-\frac{\ell}{n}\right)} \label{ineq:nswNoCase}
\end{align}
Note that the function $f(x) \coloneqq \left( \frac{2-x}{1-x} \right)^{(1-x)}$ is decreasing in the interval $x \in \left[ \frac{1}{e}, 1 \right)$ (see Claim  \ref{proposition:fndec} in Appendix \ref{appendix:apx-hard}). Hence, using the fact that $\ell \geq \frac{n}{e}$ and inequality (\ref{ineq:nswNoCase}), we get
\begin{align}
\NSW(\F)  \leq \left(\frac{2- 1/e}{1- 1/e}\right)^{1-\frac{1}{e}} \leq 1.83 \label{ineq:nswAllNo}
\end{align}
Since, in the {\rm NO} case, inequality (\ref{ineq:nswAllNo}) holds for all solutions $\mathcal{F}$, we get that the optimal NSW is at most $1.83$. 

Overall, we get that in the {\rm YES} case the optimal NSW is at least $2$ and in the {\rm NO} case it is at most $1.83$. This multiplicative gap of $\frac{2}{1.83} > 1.092$ implies that a $1.092$-approximation algorithm for NSW maximization can be used to distinguish between the two cases. Since this differentiation is {\rm NP}-hard, a $1.092$-approximation is {\rm NP}-hard as well. The theorem stands proved. 
\end{proof}

\section{Conclusion and Future Work}
The current paper extends the scope of coverage problems from combinatorial optimization to fair division. In this setting, we develop algorithmic and hardness results for maximizing the Nash social welfare. The coverage framework considered in this work accommodates expressive combinatorial constraints and, hence, it models a range of applications. The framework also generalizes public decision making among agents that have binary additive valuations. 
 
%A natural extension of the current work would be to obtain approximation guarantees for the asymmetric version of NSW: each agent $i$ has an entitlement $\eta_i \in \mathbb{R}_+$ and the goal is to find a solution $\mathcal{F}$ that maximizes $\left( \prod_i \left(v_i\left(\mathcal{F} \right)\right)^{\eta_i} \right)^{\frac{1}{\sum_i \eta_i}}$. 
 It would be interesting to extend the coverage framework to settings in which each agent $i$ has value $v^t_i$ for getting covered by the $t$th selected subset and her valuation is additive across the $T$ selections. Online version of fair coverage is another interesting direction for future work.

\bibliographystyle{alpha} 
\bibliography{references.bib}

\newpage

\appendix
\section{Missing Proofs from Section \ref{section:analysis}}
\label{appendix:prop-proofs}

Here, we restate and prove the claims used in Section \ref{section:analysis}.

\Boundphidiff*
\begin{proof}
For each agent $\ell \in [n] \setminus (X \cup F_t)$, the coverage value remains unchanged between the solutions $(X,\mathcal{F}_{-t})$ and $\mathcal{F}$, i.e., $v_\ell(X,\mathcal{F}_{-t}) = v_\ell(\mathcal{F})$. Similarly, for all $\ell \in X \cap F_t$, we have $v_\ell(X,\mathcal{F}_{-t}) = v_\ell(\mathcal{F})$. Now, considering the change in coverage values for agents in subsets $X \setminus F_t$ and $F_t \setminus X$, respectively, we obtain 
\begin{align} 
	\varphi(X,\mathcal{F}_{-t}) - \varphi(\mathcal{F}) & = \sum_{i\in X\setminus F_t} \left( \log(v_i+1) - \log(v_i) \right) - \sum_{j \in F_t\setminus X} \left(\log(v_j) - \log(v_j-1) \right) \label{eq:phidiff} 
	\end{align}
Note that, for any integer $a\geq 2$, we have $\log(a+1) - \log(a)\leq \log(a) - \log(a-1)$. We instantiate this inequality with $a = v_k$, for all agents $k \in X \cap F_t$, and extend equation (\ref{eq:phidiff}) as follows 
\begin{align}
\varphi(X,\mathcal{F}_{-t}) - \varphi(\mathcal{F}) & \geq \sum_{i\in X\setminus F_t} \left( \log(v_i+1) - \log(v_i) \right) - \sum_{i\in F_t\setminus X} \left(\log(v_i) - \log(v_i-1) \right) \nonumber \\
& \ + \sum_{k\in X \cap F_t} (\log(v_k+1) - \log(v_k)) - (\log(v_k) - \log(v_k-1)) \label{ineq:phidiff}
\end{align}
Furthermore, recall that the natural logarithm satisfies the following bounds: $\log(v +1) - \log v \geq \frac{1}{v +1}$, for all integers $v \geq 1$, and $\log v - \log(v -1) \leq \frac{1}{v-1}$, for all integers $v \geq 2$.
%https://math.stackexchange.com/questions/324345/intuition-behind-logarithm-inequality-1-frac1x-leq-log-x-leq-x-1

With these bounds, inequality (\ref{ineq:phidiff}) reduces to
\begin{align*}
		\varphi(X,\mathcal{F}_{-t}) - \varphi(\mathcal{F})
		&\geq \sum_{i\in X} \left( \log(v_i+1) - \log(v_i) \right) \ - \  \sum_{j \in F_t} \left( \log(v_j) - \log(v_j -1) \right) \\
		&\geq \sum_{i\in X} \frac{1}{v_i+1} \ - \  \sum_{j \in F_t} \frac{1}{v_j-1}.
\end{align*}
The claim stands proved. 
\end{proof}
\hfill 

Next, we state and prove Claim \ref{prop:epsilonineq}.

\begin{restatable}{claim}{Epsilonineq}\label{prop:epsilonineq}
For parameter $\varepsilon \in (0,1)$ along with any integers $\alpha\geq 4$ and $v\geq 1$,  we have
\begin{align*}
\frac{\alpha (2.25 + \varepsilon) v -1}{v+1} \geq \left(1+\frac{\varepsilon}{2}\right)\alpha.
\end{align*}
\end{restatable}
\begin{proof}
The left-hand-side of the desired inequality simplifies to 
\begin{align}
\frac{\alpha (2.25 + \varepsilon) v -1}{v+1} &  = \frac{\alpha (2.25 + \varepsilon) v}{v+1} - \frac{1}{v+1} \nonumber \\ 
& =  \alpha \left( 2 + \frac{1}{4} + \varepsilon\right) \frac{v}{v+1} \ - \ \frac{1}{v+1} \label{eq:vgeqone}
\end{align}
Since $v \geq 1$, we have $\frac{v}{v+1} \geq \frac{1}{2}$ and $\frac{1}{v+1} \leq \frac{1}{2}$. Hence, equation (\ref{eq:vgeqone}) reduces to
\begin{align*}
\frac{\alpha (2.25 + \varepsilon) v -1}{v+1} & \geq \alpha \left( 2 + \frac{1}{4} + \varepsilon\right) \frac{1}{2} - \frac{1}{2} \\
& = \alpha \left(1 + \frac{\varepsilon}{2} \right) + \frac{\alpha}{8} - \frac{1}{2} \\
& \geq \alpha \left(1 + \frac{\varepsilon}{2} \right) \tag{since $\alpha \geq 4$}
\end{align*}

This completes the proof. 
\end{proof}
\hfill 

Claim \ref{claim:WtPhi} is restated and proved next. 

\ClaimWtPhi*
\begin{proof}
Write $v_i \coloneqq v_i(\mathcal{F})$ for all agents $i \in [n]$ and note that 
\begin{align*}
\varphi(X,\mathcal{F}_{-t}) - \varphi(\mathcal{F}) = \sum_{i\in X\setminus F_t} \left( \log(v_i+1) - \log(v_i) \right)- \sum_{i\in F_t\setminus X} \left(\log(v_i) - \log(v_i-1)\right).
\end{align*}
Adding and subtracting $\sum\limits_{i\in X \cap F_t} \left(\log(v_i) - \log(v_i-1)\right)$ we get
\begin{align}
\varphi(X,\mathcal{F}_{-t}) - \varphi(\mathcal{F}) = &\sum_{i\in X\setminus F_t} \log(v_i+1) - \log(v_i) +\sum\limits_{i\in X\cap F_t} \log(v_i) - \log(v_i-1) \nonumber \\
& - \sum\limits_{i\in X\cap F_t} \log(v_i) - \log(v_i-1)  - \sum_{i\in F_t\setminus X} \log(v_i) - \log(v_i-1) \label{eq:phidiffapx}
\end{align}
Now, considering the definition of the weights $w^t_i$ (in Lines~\ref{line:weights} and \ref{line:weights1}) and equation (\ref{eq:phidiffapx}), we obtain 
\begin{align*}
\varphi(X,\mathcal{F}_{-t}) - \varphi(\mathcal{F}) = \sum_{i\in X} w^t_i - \sum_{i\in F_t} w^t_i. 
\end{align*}
This equality establishes the claim. 
\end{proof}
\hfill 

The numeric inequality used in the proof of Theorem \ref{theorem:main} is established next. 
\begin{restatable}{claim}{Dinequality}\label{prop:dineq}
For any integer $\ell\geq 2$, we have $\prod\limits_{d=2}^\ell\left(\frac{1}{2^{d-1}}\right)^\frac{1}{2^d}\geq \frac{1}{2}$.
\end{restatable}
\begin{proof}
Taking the logarithm (to the base $2$) of the left-hand-side of the desired inequality, we get
\begin{align*}
\sum\limits_{d=2}^\ell \frac{1}{2^d} \ \log\left({\frac{1}{2^{d-1}}}\right) = \sum_{d=2}^\ell \frac{-(d-1)}{2^d}.
\end{align*}
Hence, to prove the claim it suffices to show that 
\begin{align}
\sum_{d=2}^\ell \frac{-(d-1)}{2^d} \geq \log \left( \frac{1}{2} \right)= -1 \label{ineq:logform}
\end{align} 
Towards this, define $S \coloneqq \sum\limits_{d=2}^\infty \frac{-(d-1)}{2^d}$ and note that $\sum_{d=2}^\ell \frac{-(d-1)}{2^d} \geq S$. We will next show that $S = -1$. This will establish inequality (\ref{ineq:logform}) and, hence, also the claim.  

Multiplying $S$ by $2$ gives us 
\begin{align*}
2S = \sum\limits_{d=2}^\infty \frac{-(d-1)}{2^{d-1}} = \sum_{d=1}^\infty \frac{-d}{2^{d}}  = - \sum_{d=1}^\infty \frac{d}{2^{d}} = -2. 
%\frac{-1}{2} + \sum_{d=2}^\infty \frac{-d}{2^d}.
\end{align*}
%The difference 
%	\begin{align*}
%		2S - S &= -\frac{1}{2} + \sum\limits_{d=2}^\infty \frac{-d + (d-1)}{2^{d}} = -\frac{1}{2} - \sum_{d=2}^\infty \frac{1}{2^d} = -1
%	\end{align*}
Therefore, $S = -1$ and the claim follows.  	
\end{proof}

\section{Missing Proof from Section \ref{section:apx-hard}}
\label{appendix:apx-hard}

Here, we establish the proposition used in the proof of Theorem \ref{theorem:apx-hard}. 

\begin{claim}
\label{proposition:fndec}
The function $f(x)= \left(\frac{2-x}{1-x}\right)^{(1-x)}$ is decreasing in the interval $x\in [\frac{1}{e},1)$. 
\end{claim}
\begin{proof}
The function can be expressed as $f(x) = \left(\frac{2-x}{1-x}\right)^{(1-x)} = \left( 1 + \frac{1}{1-x}\right)^{(1-x)}$.  We substitute $z = 1-x$ and note that $\left(1 + \frac{1}{z}\right)^z$ is an increasing function of $z > 0$. Therefore, $f(x)$ is decreasing when $x <1$. 
\end{proof}

\section{Inapproximability in the Absence of Smoothing}
\label{appendix:non-smooth}
This section shows that that if each agent's value is equated to exactly the number of times it is covered among the subsets, then one cannot achieve any multiplicative approximation guarantee for Nash social welfare maximization. Recall that, for any solution $\mathcal{F}=(F_1, \ldots, F_T)$, the coverage value of agent $i \in [n]$ is defined as $v_i(\mathcal{F})  \coloneqq | \{t\in [T] : i \in F_t \} | + 1$. In this section, we write $c_i(\mathcal{F}) \coloneqq \left(v_i(\mathcal{F}) - 1\right)$, for all agents $i \in [n]$, i.e., $c_i(\mathcal{F})$ is equal to the number of times agent $i$ is covered under solution $\mathcal{F}$. In addition, let $\NSW^c$ denote the Nash social welfare without the smoothing, $\NSW^c(\F) \coloneqq  \left( \prod_{i=1}^n c_i(\F) \right)^{\frac{1}{n}}$. The theorem below shows that, in the absence of smoothing, the problem of maximizing NSW cannot be multiplicatively approximated.  The result is obtained via a simple reduction from the vertex cover problem and it holds for coverage instances in which the set families $\mathcal{I}_t$-s are polynomially large. 

\begin{theorem}
\label{lem:nomultapx}
Maximizing $\NSW^c$ does not admit any nontrivial multiplicative approximation guarantee, unless {\rm P} = {\rm NP}.
\end{theorem}
\begin{proof}
We will establish the theorem via a reduction from the vertex cover problem. In particular, we will show that if there exists a polynomial-time $\gamma$-approximation algorithm for maximizing $\NSW^c$, for any $\gamma < \infty$, then one can solve the {\rm NP}-complete problem of vertex cover in polynomial time.

Recall that an instance of the vertex cover problem consists of a graph $G=(V,E)$ along with a threshold $k \in \mathbb{Z}_+$ and the objective is to determine whether $G$ admits a vertex cover of size at most $k$. Given $G=(V,E)$ and $k$, we construct a coverage instance $\langle [n], T, \{\mathcal{I}_t\}_{t=1}^{T}\rangle$ with $n = |E|$ and $T = k$. That is, we associate an agent with each edge in $E$. Furthermore, for each vertex $v \in V$, we define $E_v \subseteq E$ as the subset of edges that are incident on $v$, i.e., all edges in $E_v$ are covered by vertex $v$. For each $t \in [T]$, the family $\I_t \subseteq 2^{[n]}$ is constructed as follows $\mathcal{I}_t = \left\{ E_v \right\}_{v \in V}$.
We will show that 
\begin{itemize}
\item[({\rm I})] If the given graph $G$ admits a vertex cover of size at most $k$, then in the fair coverage instance there exists a solution $\F$ with $\NSW^c(\F) \geq 1$. 
\item[({\rm II})] Otherwise, if all the vertex covers in $G$ are of size more than $k$, then $\NSW^c(\F)=0$, for all solutions $\F$. 
\end{itemize}
Here, the optimal value of $\NSW^c$ is either at least $1$ or it is $0$. Therefore, any $\gamma$-approximation algorithm (with $\gamma < \infty$) can be used to distinguish between these two cases. That is, using a $\gamma$-approximation algorithm, one can decide whether there is a vertex cover of size at most $k$ or not. This overall shows that it is {\rm NP}-hard to approximate the optimal $\NSW^c$ within any multiplicative factor. To complete the proof we will next establish properties ({\rm I}) and ({\rm II}). 

For ({\rm I}), consider a size-$k$ vertex cover $U \subseteq V$ and populate the size-$k$ tuple $\mathcal{F} = \left(E_u \right)_{u \in U}$, i.e., for each vertex $u \in U$ we include the set of covered edges $E_u$ in the tuple $\mathcal{F}$. By construction, for each vertex $v$ and every $t \in [T]$, we have $E_v \in \mathcal{I}_t$. Hence, $\mathcal{F}$ is a legitimate solution in the constructed coverage instance. Moreover, the fact that $U$ is a  vertex cover implies that every edge $e \in E$ is contained in at least one of the subsets in $\mathcal{F}$. Therefore, for all the agents $e$ (associated with the edges), we have $c_e(\mathcal{F}) \geq 1$ and $\NSW^c(\F) \geq 1$.

For ({\rm II}), assume, towards a contradiction, that there exists a solution $\F=(F_1, F_2,\ldots, F_k)$ with the property that $\NSW^c(\F) > 0$; recall that $T = k$. This bound implies that $c_e(\F) \geq 1$ for every agent $e \in [n]$. Also, note that for each $F_t \in \mathcal{I}_t$ there exists a vertex $u_t \in V$ such that $F_t = E_{u_t}$. Write $U$ to denote the subset of vertices whose incident edge sets appear in $\F$, i.e., $U  \coloneqq \{u_t \}_{t=1}^k$.  Since $c_e(\F) \geq 1$ for all agents $e \in [n]$, we get that $U$ is a vertex cover of cardinality at most $k$. This, however, contradicts the fact that, in the underlying graph $G$, all vertex covers are of size more than $k$. Hence, property ({\rm II}) holds. 

The theorem stands proved.
\end{proof}

\end{document}